\documentclass[10pt,a4paper]{amsart}
\usepackage[cp1250]{inputenc}
\usepackage[english]{babel}
\usepackage{color}
\usepackage[colorlinks,linkcolor=blue,citecolor=red]{hyperref}
\usepackage{amsaddr}
\usepackage{amssymb}
\usepackage{amsmath}
\usepackage{lmodern}
\usepackage[IL2]{fontenc}
\usepackage{url}
\usepackage{graphicx}
\usepackage{standalone}

\usepackage{pgf,tikz}
\usepackage{mathrsfs}
\usetikzlibrary{arrows}
\usetikzlibrary{snakes}
\usetikzlibrary{arrows}
\usetikzlibrary{calc}
\usetikzlibrary{patterns}
\usepackage{tkz-fct}

\newtheorem{defi}{\bf D\scriptsize EFINITION \normalsize}
\newtheorem{theorem}{\bf T\scriptsize HEOREM \normalsize}
\newtheorem{lm}{\bf L\scriptsize EMMA \normalsize}
\newtheorem{dk}{\bf C\scriptsize OROLLARY \normalsize}
\newtheorem{rem}{\bf R\scriptsize EMARK \normalsize}
\newtheorem{exa}{\bf E\scriptsize XAMPLE \normalsize}
\newtheorem{pro}{\bf P\scriptsize ROBLEM \normalsize}
\newtheorem{prop}{\bf P\scriptsize ROPOSITION \normalsize}
\newtheorem{no}{\bf N\scriptsize OTE \normalsize}

\newenvironment{remark}{\begin{rem}\rm}{\end{rem}}
\def\kopr{\hfill\raisebox{3pt}{\framebox{$\star$}}}
\newenvironment{example}{\begin{exa}\rm}{$\kopr$\end{exa}}

\newenvironment{corollary}{\begin{dk}\it}{\end{dk}}

\newcommand{\inte}[2]{\int \limits_{#1}^{#2}}

\newcommand{\zav}[1]{\left( #1 \right)}

\newcommand{\szav}[1]{\left\{ #1 \right\}}

\newcommand{\abs}[1]{\left| #1 \right|}

\newcommand{\dd}{{\rm d}}
\newcommand{\ii}{{\rm i}}
\allowdisplaybreaks
\begin{document}
\title{Pedal coordinates and free double linkage}
\author{Petr Blaschke$^1$, Filip Blaschke$^{2,3}$, Martin Blaschke$^3$}
\email{Petr.Blaschke@math.slu.cz, Filip.Blaschke@physics.slu.cz, Martin.Blaschke@physics.slu.cz}

\address{ $^1:$ Mathematical Institute, Silesian University in Opava, Na Rybn\'i\v cku 1, 746 01 Opava, Czech Republic. \\
$^2:$ Institute for Experimental and Applied Physics, Czech Technical university in Prague, Husova 240/5, 110~00 Prague 1, Czech Republic.\\ 
$^3:$ Research Centre for Theoretical Physics and Astrophysics, Institute of Physics, Silesian University in Opava, Bezru\v covo n\'am\v est\'i 1150/13, 746 01 Opava, Czech Republic.}

\thanks{To appear  in
Journal of Geometry and Physics,
2021,
104397,
ISSN 0393-0440,
\url{https://doi.org/10.1016/j.geomphys.2021.104397}.}

\begin{abstract} 
Using the technique of pedal coordinates we investigate the orbits of a free double linkage. We provide a geometrical construction for them and also show a surprising connection between this mechanical system and orbits around a Black Hole and solutions of Dark Kepler problem. 
\end{abstract}
\maketitle
\section{Introduction}
\textit{Pedal coordinates} are a largely forgotten topic in the geometry of planar curves. There is some activity in the area \cite{Pedalref1}, \cite{Pedalref2}, \cite{Pedalref3}, \cite{Pedalref4}, \cite{Pedalref5}, \cite{Pedalref6} but not as much as it deserves. As proved in \cite{Blaschke6}, pedal coordinates are particularly well suited for studying force problems within a plane. Specifically, certain class of central and Lorentz-like force problems can be easily solved in pedal coordinates algebraically (see Theorem \ref{Bl1}).  This result was further generalized for a larger set of forces in \cite{Blaschke9}. In addition, pedal coordinates  offer not only a solution but also an understanding.  They can be used for a geometrical construction of the solution as well as provide a link between seemingly unrelated problems.

The purpose of this paper is to showcase these qualities on a non-trivial example that is outside the scope of papers \cite{Blaschke6,Blaschke9} -- demonstrating that pedal coordinates are far more versatile. Though the example itself is not without interest, the emphasis is laid on the solution's method. Mainly on the ability of the method to classify the resulting trajectories.    

The problem we choose is that of movement of a \textit{free double linkage} -- i.e. a mechanical system with three point-masses, two of which are linked to the central node by massless rigid rods. One can also think of this as a double pendulum but unattached. No outside force is acting on the linkage, hence it is ``free''. Figure \ref{Obr1} illustrates the setup.  

\begin{figure}[h] 
\caption{Movement of a free double linkage corresponding to a case of equal masses and ratio of lengths 2:1. Left: Absolute frame. Right: Center of mass frame.}
\centering 
\includegraphics[scale=0.3,trim=1cm 1cm 1cm 3.5cm,clip]{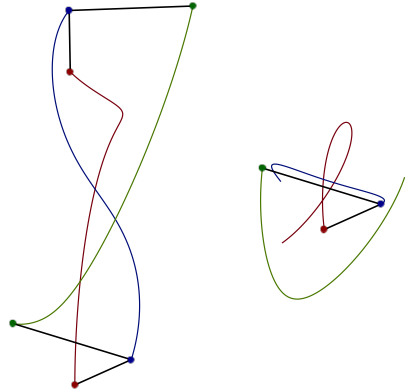}
\label{Obr1}
\end{figure}

We will consider only the planar case -- i.e. when all three points, as well as their velocity vectors, lie in a common plane. The main object of our study is the orbit of an outside node in the center of mass (CM) reference frame.

The double linkage can produce a surprisingly rich family of curves -- as demonstrated by a small gallery of cases in Figure \ref{galerydlinkage}.

\begin{figure}[h]
\caption{A gallery of trajectories of free double linkages as viewed from CM frame.  All the points have equal masses. Top row: 1:1 ratio of lengths. Bottom row: 2:1 ratio.}
\centering 
\includegraphics[scale=0.75,trim=1cm 10cm 1cm 3.5cm,clip]{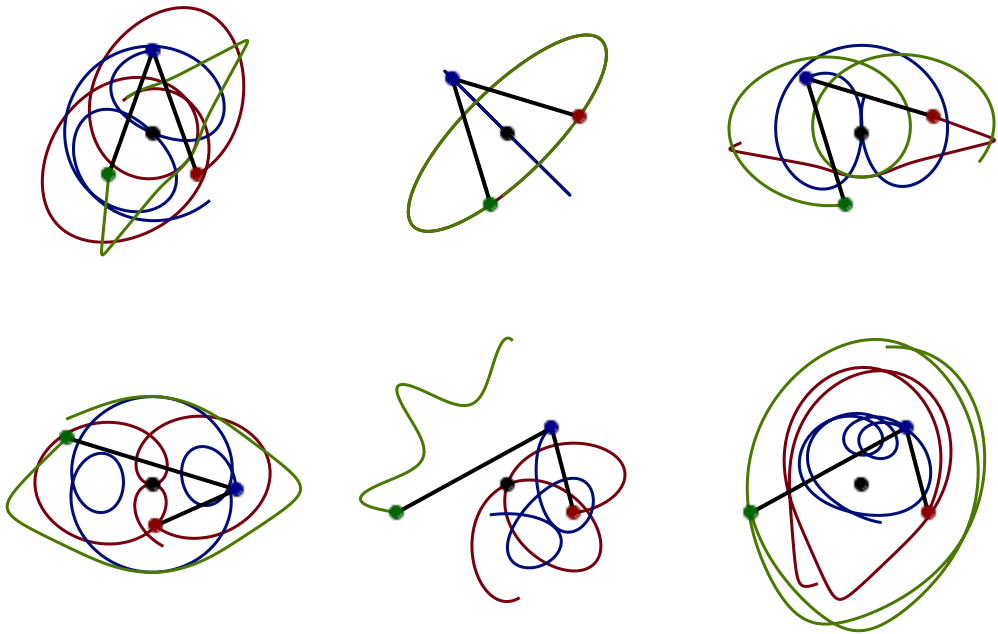}
\label{galerydlinkage}
\end{figure}

We provide pedal equation for a general case of arbitrary masses and arbitrary lengths of the rods in Theorem \ref{Th1}. In Corollary \ref{Cor3} we also show that in the special case of zero total angular momentum, the orbit of an outside node can be fully described by an interesting class of curves, which we call \textit{ellipses in generalized rotational frame of reference} (GRFR) (Figure \ref{EGFRF}).

\begin{figure}[h] 
\caption{Ellipse in a generalized rotation frame of reference}
\centering 
\includegraphics[scale=0.25,trim=5cm 1cm 10cm 1.5cm,clip]{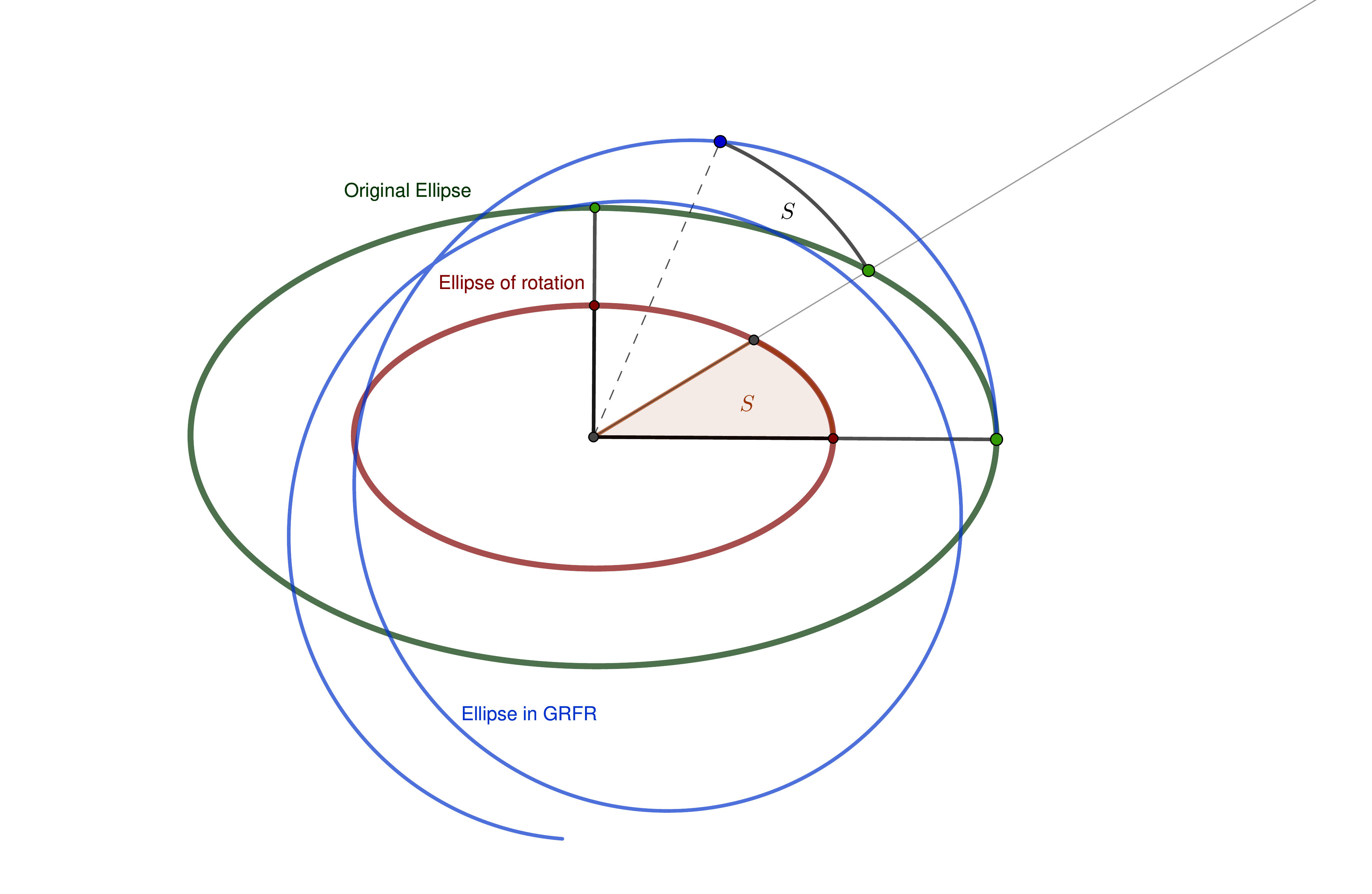}
\label{EGFRF}
\end{figure}

An (ordinary) ellipse in a rotational frame of reference (RFR) is a curve produced just by rotating every point on an ellipse by the amount of area $S$ swept by the radial vector (emanating from the center) up to that point. These curves are exactly what we get when solving Hooke's law, i.e. a dynamical system of the form 
$$
{\bf \ddot x}=-\omega {\bf x},
$$ 
in RFR. Hence the name. Remember that the trajectory of any central force problem (which Hooke's law is an example of) is parameterized so that equal areas are swept in equal time. Ellipse in GRFR is basically the same curve but with the area computed on a different co-axial ellipse (see Figure \ref{EGFRF}).

Using pedal coordinates we are also able to show a link between trajectories of a free double linkage and orbits of a test particle inside a Schwarzschild Black hole (see Section \ref{Sec3}). We have also found that the free double linkage provides a particular solution to the ``Dark Kepler problem'' (discussed in \cite{Blaschke6}). Although certainly just a mathematical coincidence, it is amazing that a connection between such dramatically different problems exists at all. This is perhaps pointing towards certain conservatism of Nature in choosing curves for its orbits.

The paper is structured as follows:  In Section \ref{Sec1}, we derive the equation of motions, conservation laws and prove Theorem \ref{Th1}. In Section \ref{Pedsec}, we give a light introduction to pedal coordinates. Finally, in Section \ref{Sec3}, we prove Corollary \ref{Cor3}.

\newpage
\section{Equation of motion and conservation laws} \label{Sec1}

Consider the dynamical system of the free double linkage, i.e. three point masses  ${\bf x},{\bf y},{\bf z}\in\mathbb{R}^3$, two of which are connected by massless rods: 
\begin{center}       

\begin{tikzpicture}[x=0.75pt,y=0.75pt,yscale=-1,xscale=1]

\draw  [fill={rgb, 255:red, 0; green, 0; blue, 0 }  ,fill opacity=1 ] (112.5,134.75) .. controls (112.5,130.75) and (115.75,127.5) .. (119.75,127.5) .. controls (123.75,127.5) and (127,130.75) .. (127,134.75) .. controls (127,138.75) and (123.75,142) .. (119.75,142) .. controls (115.75,142) and (112.5,138.75) .. (112.5,134.75) -- cycle ;
\draw    (119.75,134.75) -- (257.5,89) ;

\draw    (257.5,89) -- (420.5,201) ;

\draw  [fill={rgb, 255:red, 0; green, 0; blue, 0 }  ,fill opacity=1 ] (251.13,89) .. controls (251.13,85.48) and (253.98,82.63) .. (257.5,82.63) .. controls (261.02,82.63) and (263.88,85.48) .. (263.88,89) .. controls (263.88,92.52) and (261.02,95.38) .. (257.5,95.38) .. controls (253.98,95.38) and (251.13,92.52) .. (251.13,89) -- cycle ;
\draw  [fill={rgb, 255:red, 0; green, 0; blue, 0 }  ,fill opacity=1 ] (413.75,201) .. controls (413.75,197.27) and (416.77,194.25) .. (420.5,194.25) .. controls (424.23,194.25) and (427.25,197.27) .. (427.25,201) .. controls (427.25,204.73) and (424.23,207.75) .. (420.5,207.75) .. controls (416.77,207.75) and (413.75,204.73) .. (413.75,201) -- cycle ;
\draw [color={rgb, 255:red, 208; green, 2; blue, 27 }  ,draw opacity=1 ][line width=2.25]    (119.75,134.75) -- (158.72,121.3) ;
\draw [shift={(162.5,120)}, rotate = 520.96] [color={rgb, 255:red, 208; green, 2; blue, 27 }  ,draw opacity=1 ][line width=2.25]    (17.49,-5.26) .. controls (11.12,-2.23) and (5.29,-0.48) .. (0,0) .. controls (5.29,0.48) and (11.12,2.23) .. (17.49,5.26)   ;

\draw [color={rgb, 255:red, 208; green, 2; blue, 27 }  ,draw opacity=1 ][line width=2.25]    (257.5,89) -- (220.29,101.71) ;
\draw [shift={(216.5,103)}, rotate = 341.15] [color={rgb, 255:red, 208; green, 2; blue, 27 }  ,draw opacity=1 ][line width=2.25]    (17.49,-5.26) .. controls (11.12,-2.23) and (5.29,-0.48) .. (0,0) .. controls (5.29,0.48) and (11.12,2.23) .. (17.49,5.26)   ;

\draw [color={rgb, 255:red, 208; green, 2; blue, 27 }  ,draw opacity=1 ][line width=2.25]    (257.5,89) -- (313.19,126.76) ;
\draw [shift={(316.5,129)}, rotate = 214.14] [color={rgb, 255:red, 208; green, 2; blue, 27 }  ,draw opacity=1 ][line width=2.25]    (17.49,-5.26) .. controls (11.12,-2.23) and (5.29,-0.48) .. (0,0) .. controls (5.29,0.48) and (11.12,2.23) .. (17.49,5.26)   ;

\draw [color={rgb, 255:red, 208; green, 2; blue, 27 }  ,draw opacity=1 ][line width=2.25]    (420.5,201) -- (363.8,162.26) ;
\draw [shift={(360.5,160)}, rotate = 394.35] [color={rgb, 255:red, 208; green, 2; blue, 27 }  ,draw opacity=1 ][line width=2.25]    (17.49,-5.26) .. controls (11.12,-2.23) and (5.29,-0.48) .. (0,0) .. controls (5.29,0.48) and (11.12,2.23) .. (17.49,5.26)   ;

\draw (128,164) node  [align=left] {$\displaystyle \mathbf{x}$};
\draw (255,118) node  [align=left] {$\displaystyle \mathbf{y}$};
\draw (386,218) node  [align=left] {$\displaystyle \mathbf{z}$};
\draw (132,99) node  [align=left] {$\displaystyle F_{1}$};
\draw (222,70) node  [align=left] {$\displaystyle F_{1}$};
\draw (307,84) node  [align=left] {$\displaystyle F_{2}$};
\draw (412,158) node  [align=left] {$\displaystyle F_{2}$};
\draw (199,129) node  [align=left] {$\displaystyle l_{1}$};
\draw (326,164) node  [align=left] {$\displaystyle l_{2}$};

\end{tikzpicture}

\end{center}
As there are no outside forces acting on the system, the equations of motion are of the form:
\begin{align*}
m_x {\bf \ddot x}&=F_1({\bf y}-{\bf x}),\\
m_y {\bf \ddot y}&=-F_1({\bf y}-{\bf x})+F_2({\bf z}-{\bf y}),\\
m_z {\bf \ddot z}&=-F_2({\bf z}-{\bf y}),\\
\end{align*}
where the internal forces $F_1,F_2$ are such that the length of the massless rods remains constant:
\begin{equation}\label{constrains}
\abs{{\bf x-\bf y}}=l_1,\qquad \abs{{\bf z-\bf y}}=l_2.
\end{equation}

It is easy to see that the\textit{ total kinetic energy} $K$
$$
K:=\frac12\zav{m_x\abs{{\bf \dot x}}^2+m_y\abs{{\bf \dot y}}^2+m_z\abs{{\bf \dot z}}^2},
$$
is conserved:
$$
\dot K=m_x {\bf \ddot x}\cdot {\bf \dot x}+m_y {\bf \ddot y}\cdot {\bf \dot y}+m_z {\bf \ddot z}\cdot {\bf \dot z}=
-F_1({\bf y}-{\bf x})\cdot ({\bf \dot y}-{\bf \dot x})-F_2({\bf z}-{\bf y})\cdot ({\bf \dot z}-{\bf \dot y})=0,
$$
where the last equality is a differential consequence of (\ref{constrains}).

By summing all equations of motion we find
$$
m_x{\bf \ddot x}+m_y{\bf \ddot y}+m_z{\bf \ddot z}=0,\qquad \Rightarrow \qquad {\bf \ddot T}=0,
$$
that the center of mass ${\bf T}$:
$$
{\bf T}:=\frac{m_x{\bf x}+m_y{\bf y}+m_z{\bf z}}{m_x+m_y+m_z},
$$
is moving uniformly in a straight line.

We can also say that the \textit{total linear momentum} ${\bf \dot T}$, is conserved.

Another conserved vector is the \textit{total angular momentum} 
$$
{\bf L}:=m_x{\bf x}\times {\bf \dot x}+m_y{\bf y}\times {\bf \dot y}+m_z{\bf z}\times {\bf \dot z}.
$$
$$
{\bf \dot L}=m_x{\bf x}\times {\bf \ddot x}+m_y{\bf y}\times {\bf \ddot y}+m_z{\bf z}\times {\bf \ddot z}=F_1({\bf x}-{\bf y})\times ({\bf y}-{\bf x})+F_2({\bf y}-{\bf z})\times({\bf z}-{\bf y})=0.
$$

Thus we have altogether 9 conserved quantities:
$$
l_1\ ,l_2\ ,K\ ,{\bf \dot T}, {\bf L}.
$$

Differential consequences of (\ref{constrains}) are: 
$$
({\bf y}-{\bf x})\cdot ({\bf \dot y}-{\bf \dot x})=0,\qquad ({\bf z}-{\bf y})\cdot ({\bf \dot z}-{\bf \dot y})=0,
$$
which tells us that the displacement vectors ${\bf u}:={\bf y}-{\bf x}$, ${\bf v}:={\bf z}-{\bf y}$ between points ${\bf y}, {\bf x}$ a ${\bf z}, {\bf y}$ are orthogonal to their respective velocities:
$$
{\bf u}\cdot {\bf \dot u}=0,\qquad {\bf v}\cdot {\bf \dot v}=0.
$$

Differentiating (\ref{constrains}) second time we get consequences from which it is possible to extract the internal forces $F_1, F_2$:
\begin{align*}
&0=\abs{{\bf \dot y}-{\bf \dot x}}^2+({\bf y}-{\bf x})\cdot ({\bf \ddot y}-{\bf \ddot x}), & &0=\abs{{\bf \dot z}-{\bf \dot y}}^2+({\bf z}-{\bf y})\cdot ({\bf \ddot z}-{\bf \ddot y}),\\
&\abs{{\bf \dot y}-{\bf \dot x}}^2=F_1 \frac{m_x+m_y}{m_x m_y}l_1^2+F_2 l_1 l_2 \frac{\cos\omega}{m_y}, &  &\abs{{\bf \dot z}-{\bf \dot y}}^2=F_2 \frac{m_z+m_y}{m_z m_y}l_2^2+F_1 l_1 l_2 \frac{\cos\omega}{m_y},\\
&F_1=\frac{m_y}{m_z }\frac{\tilde m_z l_2^2\abs{{\bf \dot y}-{\bf \dot x}}^2-\cos\omega l_1 l_2 m_z \abs{{\bf \dot z}-{\bf \dot y}}^2}{l_1^2 l_2^2\zav{\frac{\tilde m_x \tilde m_z}{m_x m_z}-\cos^2\omega}},& &F_2=\frac{m_y}{m_x}\frac{\tilde m_x l_1^2\abs{{\bf \dot z}-{\bf \dot y}}^2-\cos\omega l_1 l_2 m_x \abs{{\bf \dot y}-{\bf \dot x}}^2}{l_1^2 l_2^2\zav{\frac{\tilde m_x \tilde m_z}{m_x m_z}-\cos^2\omega}},
\end{align*}
where, for the sake of brevity, we have introduced the notation:
$$
\tilde m_x:=m_x+m_y,\qquad \tilde m_z:=m_z+m_y,
$$
and $\omega$ is the angle between the two rods:
$$
\cos\omega:=\frac{({\bf x-\bf y})\cdot ({\bf z-\bf y})}{l_1 l_2}.
$$
\subsection{Center of mass frame of reference}

Without the loss of generality we can move to an inertial frame of reference in which the center of mass is at rest at the origin:
\begin{align*}
{\bf T}&=0, & &\Rightarrow & {\bf y}&=-\frac{m_x{\bf x}+m_z{\bf z}}{m_y},\\
{\bf \dot T}&=0, & &\Rightarrow & {\bf \dot y}&=-\frac{m_x{\bf \dot x}+m_z{\bf \dot z}}{m_y}.\\
\end{align*}

The equations of motion become: 
\begin{align*}
m_x{\bf \ddot x}&=-F_1\frac{\tilde m_x{\bf x}+m_z {\bf z}}{m_y}, & F_1&=\frac{\tilde m_z l_2^2\abs{\tilde m_x{\bf \dot x}+m_z{\bf \dot z}}^2-\cos\omega l_1 l_2 m_z \abs{\tilde m_z{\bf \dot z}+m_x{\bf \dot x}}^2}{m_y m_z l_1^2l_2^2(\beta_0-\cos^2\omega)},\\
m_z {\bf \ddot z}&=F_2\frac{\tilde m_z{\bf z}+m_x {\bf x}}{m_y}, & F_2&=\frac{\tilde m_x l_1^2\abs{ m_x{\bf \dot x}+\tilde m_z{\bf \dot z}}^2-\cos\omega l_1 l_2 m_x \abs{ m_z{\bf \dot z}+\tilde m_x{\bf \dot x}}^2}{m_y m_x l_1^2l_2^2(\beta_0-\cos^2\omega)},\\
\cos\omega&=\frac{(\tilde m_z {\bf z}+ m_x {\bf x})\cdot (m_z {\bf z}+ \tilde m_x {\bf x})}{m_y^2 l_1 l_2}.
\end{align*}
Our conserved quantities are now in the form:
\begin{equation}\label{constrains2}
\abs{\tilde m_x {\bf x}+m_z {\bf z}}=l_1 m_y,\qquad \abs{m_x {\bf x}+\tilde m_z {\bf z}}=l_2 m_y.
\end{equation}

\begin{equation}\label{Kenergy}
K=\frac{\tilde m_x m_x}{2m_y} \abs{{\bf \dot x}}^2+\frac{\tilde m_z m_z}{2m_y} \abs{{\bf \dot z}}^2+\frac{m_x m_z}{m_y} ({\bf \dot x}\cdot {\bf \dot z}).
\end{equation}
\begin{equation}\label{Angmomentum}
 {\bf L}=\frac{\tilde m_x m_x}{m_y}({\bf x}\times {\bf \dot x})+\frac{\tilde m_z m_z}{m_y}({\bf z}\times {\bf \dot z})+\frac{m_x m_z}{m_y}({\bf x}\times {\bf \dot z}+{\bf z}\times {\bf \dot x}).
\end{equation}

\subsection{Coplanar case}
With a loss of generality, we now focus on the case, where the movement of  the whole system is confined to a plane.  Let us define the following notation: 
\begin{align*}
{\bf x}^\perp&:=\frac{({\bf x}\times {\bf \dot x})\times {\bf x}}{\abs{{\bf x}\times {\bf \dot x}}}=\frac{{\bf \dot x}\abs{{\bf x}}^2-{\bf x}({\bf x}\cdot {\bf \dot x})}{\abs{{\bf x}\times {\bf \dot x}}}, & {\bf \dot x}^\perp&:=\frac{({\bf x}\times {\bf \dot x})\times {\bf \dot x}}{\abs{{\bf x}\times {\bf \dot x}}}=\frac{-{\bf x}\abs{{\bf \dot x}}^2+{\bf \dot x}({\bf x}\cdot {\bf \dot x})}{\abs{{\bf x}\times {\bf \dot x}}}.
\end{align*}
The vector ${\bf x}^\perp$ is perpendicular to ${\bf x}$ and clearly lies in the plane to which belong both ${\bf x}$ and ${\bf \dot x}$. And similarly for ${\bf \dot x}^\perp$.

It is a stimulating exercise to verify the following properties:
\begin{align*}
&{\bf x}^\perp \cdot {\bf x}=0, & &\abs{{\bf x}^\perp}=\abs{{\bf x}},\\
&{\bf \dot x}^\perp \cdot {\bf \dot x}=0, & &\abs{{\bf \dot x}^\perp}=\abs{{\bf \dot x}},\\
&{\bf x}^\perp \cdot {\bf \dot x}=\abs{{\bf x}\times {\bf \dot x}}, & &{\bf x}^\perp \cdot {\bf \dot x}^\perp={\bf x} \cdot {\bf \dot x},\\
&{\bf x} \times {\bf \dot x}^\perp=\frac{({\bf x}\cdot {\bf \dot x})}{\abs{{\bf x}\times {\bf \dot x}}} {\bf x}\times {\bf \dot x}, & &{\bf x}^\perp \times {\bf \dot x}=-\frac{({\bf x}\cdot {\bf \dot x})}{\abs{{\bf x}\times {\bf \dot x}}}{\bf x}\times {\bf \dot x},\\
&{\bf x}^\perp \times {\bf \dot x}^\perp ={\bf x}\times {\bf \dot x}.
\end{align*}

We are considering the coplanar case specifically for the possibility to use the so called {pedal coordinates} $r,p$, which are a special type of coordinates that do not determine the position of each point on a curve (as Cartesian or polar coordinates do) but, instead, measures the distance of a point from the origin ($r$) and { the distance of a tangent in that point to the origin} ($p$). This clearly shows that the pedal coordinates can describe planar curves only. They are not general purpose coordinates.

Specifically, let us define for every point ${\bf x}$ and its tangent vector ${\bf \dot x}$ the following quantities:  
\begin{align*}
r&:=\abs{{\bf x}}, & \text{ The distance form the origin.}\\
p&:=\frac{{\bf x}^\perp \cdot {\bf \dot x}}{\abs{{\bf \dot x}}}, & \text{The oriented distance of the tangent from the origin.}\\
p_c&:=\frac{{\bf x} \cdot {\bf \dot x}}{\abs{{\bf \dot x}}}, & \text{The oriented distance of the normal from the origin.}\\
\kappa&:=\frac{{\bf \dot x}^\perp\cdot {\bf \ddot x}}{\abs{{\bf \dot x}}^3}, &\text{The oriented curvature.}
\end{align*} 
The so-called contrapedal coordinate $p_c$ is {not} an independent variable, but in some cases it is more convenient to use. Its relation with $(r,p)$ is as follows:
$$
p^2+p_c^2=r^2.
$$
One of the many advantages of pedal coordinates is that computation of the curvature is much more simple.
$$
\frac{1}{r}\frac{\dd p}{\dd r}=\kappa.
$$

In terms of pedal coordinates we are able to prove the following:
\begin{theorem}\label{Th1} The curve traced by the node ${\bf x}$ of a coplanar free double linkage as viewed from the center of mass is given by the \textit{algebraic} pedal equation:
\begin{equation}\label{Pedeq}
\zav{\abs{\bf L}^2-2L_1\zav{r^2}K r^2}\zav{L_1\zav{r^2}\frac{p\beta}{p_c}+L_2\zav{r^2}}^2+\abs{\bf L}^2 P_2\zav{r^2}=0,
\end{equation}
where $M:=m_x+m_y+m_z$ is total mass and
\begin{align}
L_1\zav{r^2}&:=\frac{M}{4\tilde m_z m_z}\zav{b+\frac{a}{r^2}},\\
L_2\zav{r^2}&:=\frac{aM^2m_y}{8m_z^2\tilde m_z^2 r^2}\zav{1+\frac{r_0r_1}{r^2}},\\
P_2\zav{r^2}&=4\frac{m_x^2M^2}{\tilde m_z^2}\beta^2+\frac{m_y^3 m_x M^5}{4m_z^3\tilde m_z^4}\frac{\zav{r_1^2-r_0^2}^2}{r^4}.\\
a&=M m_y(r_1-r_0)^2,\qquad b=4m_xm_z,\\
\label{distances} r_{1}&=\frac{l_1\tilde m_z+ l_2m_z}{M},\qquad r_{0}=\frac{l_1\tilde m_z- l_2m_z}{M}\\
\beta&:=\frac{m_y M}{2  m_z\tilde m_z  r^2}\sqrt{\zav{r_1^2- r^2}\zav{r^2-r_0^2}}.
\end{align}
\end{theorem}
\begin{proof}
Since ${\bf z}, {\bf \dot z}$ are in the same plane as ${\bf x},{\bf \dot x}$, we can write:
$$
{\bf z}=\alpha {\bf x}+\beta {\bf x}^\perp,\qquad {\bf \dot z}=a {\bf \dot x}+b {\bf \dot x}^\perp,
$$
for some scalar functions $\alpha,\beta,a,b$.

There are exactly 12 ways how to write ${\bf z},{\bf \dot z}$ as a linear combination of two base vectors chosen from the set $\szav{{\bf x},{\bf x}^\perp,{\bf \dot x},{\bf \dot x}^\perp}$. But only one allows us to compute the coefficients $\alpha,\beta,a,b$ in geometrical terms, as we will see.

Substituting ${\bf z}$ into (\ref{constrains2}) we get
$$
\zav{\tilde m_x+m_z\alpha }^2+m_z^2\beta ^2=\frac{l_1^2 m_y^2}{r^2},\qquad
\zav{ m_x+\tilde m_z\alpha }^2+{\tilde m_z}^2\beta ^2=\frac{l_2^2 m_y^2}{r^2},
$$ 
which yields:
$$
\alpha=\alpha_0+\frac{\alpha_1}{r^2},\qquad  \alpha^2+\beta^2=\beta_0+\frac{\beta_1}{r^2},
$$
where
$$
\alpha_0:=-\frac{\tilde m_z\tilde m_x+m_z m_x}{2m_z\tilde m_z},\qquad \alpha_1:=\frac{l_1^2 {\tilde m_z}^2-l_2^2 m_z^2}{2m_z\tilde m_z (m_x+m_y+m_z)}m_y,
$$
$$
\beta_0:=\frac{\tilde m_x  m_x}{\tilde m_z m_z},\qquad
\beta_1:=\frac{l_2^2 \tilde m_x m_z-l_1^2 m_x\tilde m_z}{m_z\tilde m_z(m_x+m_y+m_z)}m_y.
$$

With this information, we can write the quantity $\cos \omega$ as follows:
$$
\cos \omega=\frac{\tilde m_z m_z}{m_y^2 l_1 l_2}\zav{\alpha^2+\beta^2-2\alpha\alpha_0+\beta_0} r^2.
$$
Now, since this expression involves only variable $r$ and since $\abs{\cos \omega}\leq 1$, this gives us ultimately a restriction on values of $r$ in the form:
$$
r_0^2\leq r^2\leq r_1^2,
$$
see (\ref{distances}).
This simply reflect the fact, that point ${\bf x}$ cannot wander away from the center of mass arbitrarily far.

Armed with this information, we can write the quantity $\beta$ in a more illuminating manner:
$$
\beta^2 =\frac{m_y^2 M^2}{4  m_z^2\tilde m_z^2  r^4}\zav{r_1^2- r^2}\zav{r^2-r_0^2}.
$$

Substituting both ${\bf z},{\bf \dot z}$ into the differential consequence of (\ref{constrains2}) we get:
$$
\zav{(\tilde m_x+\alpha m_z) {\bf x}+\beta m_z {\bf x}^\perp}\cdot \zav{(\tilde m_x+a m_z) {\bf \dot x}+b m_z {\bf \dot x}^\perp}=0, 
$$
$$
\zav{(m_x+\alpha \tilde m_z) {\bf x}+\beta \tilde m_z {\bf x}^\perp}\cdot \zav{( m_x+a \tilde m_z) {\bf \dot x}+b \tilde m_z {\bf \dot x}^\perp}=0,
$$
which using properties of ${\bf x}^\perp,{\bf \dot x}^\perp$ translate into pedal coordinates as follows:
$$
((\tilde m_x+\alpha m_z)(\tilde m_x+a m_z)+\beta b m_z^2)p_c+(\beta m_z (\tilde m_x+a m_z)-b m_z (\tilde m_x+\alpha m_z))p=0,
$$
$$
(( m_x+\alpha \tilde m_z)(m_x+a \tilde m_z)+\beta b {\tilde m_z}^2)p_c+(\beta \tilde m_z ( m_x+a \tilde m_z)-b \tilde m_z (m_x+\alpha \tilde m_z))p=0.
$$
This is a linear system in $a,b$ and its solution is:
\begin{align*}
a&=\frac{(\alpha p-\beta p_c)((\alpha-2\alpha_0)p_c+\beta p)+\beta_0 p p_c}{\beta r^2},\\
b&=\frac{(\alpha p_c+\beta p)((\alpha-2\alpha_0)p_c+\beta p)+\beta_0 p_c^2}{\beta r^2}.
\end{align*}

Interesting consequences of these identities that will be useful momentarily are 
$$
p b-p_c a=(\alpha-2\alpha_0)p_c+\beta p,\qquad a p+b p_c=\frac{\alpha}{\beta} ((\alpha-2\alpha_0)p_c+\beta p)+\frac{\beta_0 p_c}{\beta}.
$$

Using the properties of ${\bf x}^\perp, {\bf \dot x}^\perp$ we can also translate our remaining conserved quantities to pedal coordinates:
\begin{equation}\label{Kenergycopl}
K=\frac{\tilde m_z m_z}{2m_y}\zav{\beta_0+a^2+b^2+\frac{2m_x}{\tilde m_z} a} \abs{{\bf \dot x}}^2.
\end{equation}
\begin{equation*}
 {\bf L}=\frac{\tilde m_z m_z}{m_y}\zav{\beta_0 p+\alpha (a p+b p_c)+\beta (bp-ap_c)+\frac{m_x}{\tilde m_z}\zav{(\alpha+a)p+(b-\beta)p_c}}\frac{\abs{{\bf \dot x}}}{\abs{{\bf x}\times {\bf \dot x}}} {\bf x}\times {\bf \dot x}.
\end{equation*}
Substituting everything we arrive at this rather simple expression:
\begin{equation}\label{Angmomentumcopl}
\frac{\abs{{\bf L}}}{\abs{{\bf \dot x}}}=\frac{p_c }{\beta}\zav{ L_1\zav{r^2}\frac{p\beta}{p_c}+L_2\zav{r^2}}.
\end{equation}

The indices of polynomials $L_1,L_2$ are chosen so that they indicate their degree, i.e. $L_1$ is a first degree polynomial in $1/r^2$, and $L_2$ is of second degree. 

The kinetic energy cannot be approached so directly. We must first bring it into the following form:
\begin{align*}
\frac{2m_y}{m_z\tilde m_z}\frac{K}{\abs{{\bf \dot x}}^2}&=
\zav{\zav{\frac{\alpha_1^2}{r^4}-\alpha_0^2+\beta_0} \frac{p_c}{\beta r}+\zav{\frac{\alpha_1}{r^2}+\alpha_0+\frac{m_x}{\tilde m_z}}\frac{p}{r}}^2\\
& +
\zav{ \zav{\frac{\alpha_1}{r^2}-\alpha_0-\frac{m_x}{\tilde m_z}} \frac{p_c}{r}+\beta \frac{p}{r}}^2
+ \frac{p_c^2+p^2}{r^2}\zav{\beta_0-\frac{m_x^2}{\tilde m_z^2}}.\nonumber
\end{align*}
Now we can see that $K$ can be written as follows:
\begin{equation}\label{kinenergycopl}
\frac{K}{\abs{{\bf \dot x}}^2}=\frac{p_c^2}{r^2\beta^2}\zav{\frac12 L1\zav{r^2} \zav{\frac{p\beta}{p_c}}^2+L_2\zav{r^2}\frac{p\beta}{p_c}+K_3\zav{r^2}},
\end{equation}
where
\begin{align*}
K_3\zav{r^2}&:=\frac{m_y^2 M^2}{8\tilde m_z^3 m_z^2}\zav{-m_x M\zav{1-\frac{r_0^2}{r^2}}\zav{1-\frac{r_1^2}{r^2}}+\frac{l_2^2m_zm_y}{r^2}\zav{1-\frac{r_0r_1}{r^2}}^2}.
\end{align*}
Again, the index indicates the degree in $1/r^2$.

Finally, eliminating the factor $\abs{{\bf \dot x}}^2$ form both $K$ and $\abs{{\bf L}}^2$, we end up with the equation (\ref{Pedeq}), which is what we want. The polynomial $P_2$ is obtained as $P_2:=2K_3 L_1-L_2^2$ and turn outs to be just of second degree.
\end{proof}
\section{Pedal coordinates}\label{Pedsec}
To understand what this result is telling us, we must first better understand the pedal coordinates and their many advantages. 

Let us introduce again for any point ${\bf x}\in\mathbb{R}^2$ on a plane curve $\gamma$ the following symbols: 
\begin{align*}
{\bf x}&:=(x,y), & {\bf x}^\perp&:=(-y,x), & {\bf x}\cdot {\bf y}^\perp&=-{\bf x}^\perp \cdot {\bf y}, & \zav{{\bf x}^\perp}^\perp&=-{\bf x}. \\
p&:=\frac{{\bf x}^\perp \cdot {\dot {\bf x}}}{\abs{\dot {\bf x}}}, &  p_c&:=\frac{{\bf x} \cdot {\dot {\bf x}}}{\abs{\dot {\bf x}}}, & \kappa&:=\frac{{\dot {\bf x}}^\perp \cdot {\ddot {\bf x}} }{\abs{\dot {\bf x}}^3}, & p^2+p_c^2&=r^2.\\
\end{align*}

This diagram roughly illustrates the definition of pedal coordinates $p,p_c,r$:
\begin{center}
{\small
\begin{tikzpicture}[domain=-4:5]
\draw [thick,color=blue] (-4,3) to [out=0,in=180] (0,2) to [out=0, in=225] (2,3);
\coordinate [label=right:{$\gamma$}] (M1) at (2,3);
\coordinate [label=above:{${\bf x}$}] (M2) at (0,2);
\draw [thin] (-4,2) to (2,2);
\draw [dashed] (-3,2) to (-3,0);
\draw [dashed] (0,2) to (-3,0);
\filldraw (-3,0) circle (0.05);
\filldraw (0,2) circle (0.05);
\filldraw (-3,2) circle (0.05);
\coordinate [label=below:{$0$}] (M3) at (-3,0);
\coordinate [label=above:{$P({\bf x})$}] (M4) at (-3,2);
\draw [densely dotted] (0,2) to (0,0);
\draw [densely dotted] (-3,0) to (0,0);
\filldraw (0,0) circle (0.05);
\coordinate [label=left:{$p$}] (M6) at (-3,1);
\coordinate [label=below right:{$r$}] (M7) at (-1.5,1);
\coordinate [label=below right:{$p_c$}] (M8) at (-1.5,0);
\end{tikzpicture}
}
\end{center}

But with $p,p_c$ we must remember that these are signed distances, which can be negative. Some of the simplest curves are given in pedal coordinates as follows:
\begin{align*}
p&=a, &\text{line distant $a$.} \\
r&=a, &\text{point distant $a$.} \\
p_c&=a, &\text{Involute of a circle.} \\
2pR&=r^2+R^2-a^2, &\text{circle of radius $R$ and center distant $a$.}
\end{align*}
The involute of a circle is not usually regarded as a ``simple curve'' but in pedal coordinates it has an extremely simple equation so perhaps it is.

\begin{remark}
Pedal coordinates do not tell us everything about the curve and they actually describe many curves at once -- if you choose to differentiate between them.

The equation $p=a$ is valid for any line distant $a$ and $r=a$ for any point distant $a$, etc. Obviously, the pedal coordinates do not care about the rotation around the pedal point and about the curve's parametrization, but it is actually not easy to tell in general the nature of ambiguity associated to a pedal equation -- in fact, it differs from equation to equation. (For more information see \cite{Blaschke6}.)

This can be seen, actually, as an advantage of pedal coordinates over other systems if you are interested only in the general shape of the curve and do not want to be distracted by other details. 
\end{remark}

\subsection{Force problems}
Perhaps the best feature of pedal coordinates is that they are particularly suited for describing orbits of force problems in classical mechanics  (\cite{Blaschke6}). Specifically, the following holds: 
\begin{theorem} \label{Bl1} (From \cite{Blaschke6}.)
Consider a dynamical system:
\begin{equation}\label{dynsys}
{\bf \ddot x}=F^\prime\zav{\abs{{\bf x}}^2}{\bf x}+2 G^\prime\zav{\abs{{\bf x}}^2}{\bf \dot x}^\perp,
\end{equation}
describing an evolution of a test particle (with position ${\bf x}$ and velocity ${\bf \dot x}$) in the plane in the presence of a central $F$ and a Lorentz like $G$ potential. The quantities:
$$
L={\bf x}\cdot {\bf \dot x}^\perp+G\zav{\abs{{\bf x}}^2}, \qquad c=\abs{{\bf\dot x}}^2-F\zav{\abs{{\bf x}}^2},
$$
are conserved in this system.

Then the curve traced by ${\bf x}$ is given in pedal coordinates by
\begin{equation}
\frac{\zav{L-G(r^2)}^2}{p^2}=F(r^2)+c,
\end{equation}
with the pedal point at the origin. 
\end{theorem}

Notice that only two conservation laws are required to obtain the shape of the orbit (in pedal coordinates) -- energy and angular momentum. These are exactly the same conservation laws we have used to solve the double linkage.

\begin{example}

Theorem \ref{Bl1} can be used to interpret a large family of pedal equations. For instance, the above result tells us that the orbits of the Kepler problem:
$$
{\bf \ddot x}=-\frac{GM}{r^3}{\bf x},
$$ 
which are of course conic sections with the origin at a focus, are given in pedal coordinates as
$$
\frac{L^2}{p^2}=\frac{2G M}{r^2}+c.
$$
We can rewrite this as follows:
\begin{equation}\label{focalconic}
r_0r_1\zav{\frac{1}{p^2}-\frac{1}{r^2}}=-\frac{(r-r_0)(r-r_1)}{r^2},
\end{equation}
where the roots $r_0,r_1$ satisfy:
$$
r_0+r_1=-\frac{2GM}{c},\qquad r_0r_1=-\frac{L^2}{c}.
$$

Since $|p|<r$, i.e. distance to a tangent is smaller than distance to a point, we have for every point on a curve:
\begin{equation}\label{basicine}
\frac{1}{p^2}-\frac{1}{r^2}\geq 0.
\end{equation}

This inequality implies a restriction for possible values of $r$ in the form:
$$
-\frac{(r-r_0)(r-r_1)}{r_0r_1 r^2}\geq 0.
$$ 

If $ 0<r_0<r_1$ this translates as $r_0\leq r \leq r_1$, and the curve is therefore bounded -- i.e. we have an ellipse with $r_0$ being the periapsis and $r_1$ the apoapsis.

If $r_0<0, r_1>0$ we have $r\geq r_1$ an unbounded curve, i.e. a hyperbola.

The case $r_0,r_1<0$ is impossible. There would be no allowed values of $r$. 
\end{example}
\begin{example}
Similarly, solutions of a dynamical system of the form
$$
{\bf \ddot x}= -\omega{\bf x},
$$ 
are very easy to compute and they are ellipses for $\omega>0$ but viewed from the center rather than from the focus.

In pedal coordinates we have thus an equation for a \textit{central ellipse}  given by:
$$
\frac{L^2}{p^2}= -\omega r^2+c,
$$
or 
\begin{equation}\label{centralconic}
a^2b^2\zav{\frac{1}{p^2}-\frac{1}{r^2}}=-\frac{(r^2-a^2)(r^2-b^2)}{r^2},
\end{equation}
where the roots $a$, $b$, given by
$$
a+b=\frac{c}{\omega},\qquad ab=\frac{L^2}{\omega},
$$
are the semi-major and the semi-minor axis respectively.

From (\ref{basicine}) we can one again infer that (given $a<b$) we have
 $a\leq r\leq b$. 
\end{example}
\begin{example}
Finally, take a problem of determining the orbit of a charged particle in a uniform magnetic field:
$$
{\bf \ddot x}= 2a {\bf \dot x}^\perp,
$$
(that is experiencing the Lorentz force only). It is well known that solutions are circles. In pedal coordinates we thus have: 
$$
\frac{\zav{L-a r^2}^2}{p^2}=c\qquad \Rightarrow\qquad \sqrt{c}p= a r^2-L.
$$
\end{example}

Pedal coordinates are actually not limited by equations of the form (\ref{dynsys}) -- as demonstrated by our very own example of free double linkage. Many more dynamical systems and even problems of calculus of variation can be readily translated into pedal coordinates. For detail see \cite{Blaschke9}. 
\subsection{Transforms} The second feature of pedal coordinates is that they are well suited for performing many \textit{transformations of curves} and thus make possible to connect seemingly unrelated problems (as we will see).

We will mention just a few examples that will be useful for understanding the movement of double linkage. For a more comprehensive lists see \cite{Blaschke6}, \cite{Lawrence}, \cite{Yates}, \cite{Lockwood},\cite{Hilton} as well as the web pages \cite{krivkystranky}. 

\subsubsection{Scaling} From definitions of quantities $p,p_c,r$ it is clear that a classical radial scaling that transforms a given point ${\bf x}$ to the point
$$
S_\alpha:\qquad {\bf \tilde x}=\frac{{\bf x}}{\alpha},
$$
is given in pedal coordinates as follows:
$$
f(p,r,p_c)=0 \qquad \stackrel{S_\alpha}{\longrightarrow}\qquad f(\alpha p, \alpha r,\alpha p_c)=0.
$$
\subsubsection{Pedal} Pedal transform maps any point on a given curve ${\bf x}$ to the orthogonal projection ${\bf \tilde x}$ of the origin to the tangent at ${\bf x}$. In the diagram at the beginning of the section, ${\bf \tilde x}$ is denoted by $P({\bf x})$.

Algebraically, the new point is given by 
$$
P:\qquad {\bf \tilde x}={\bf x}-p_c \frac{\bf \dot x}{\abs{\bf \dot x}}. 
$$

This transformation has very nice properties. It maps a focal ellipse or a hyperbola (that is with the origin at the focus) to its circumcircle. The pedal of a parabola is a line  -- not the directrix, though; this line is parallel to the directrix and passes through the vertex. A central rectangular hyperbola is mapped to the Lemniscate of Bernoulli and so on. 

In fact, pedal coordinates was originally devised (as the name suggests) to make the pedal transformation easy to do. In Cartesian coordinates it is generally required to solve a differential equation. But the same thing can be by done algebraically in pedal coordinates.

Specifically (see \cite[p. 228]{williamson}), to any curve given by the equation
$$
f(p,r, p_c)=0,
$$ 
the pedal curve satisfies the equation
$$
f\zav{r,\frac{r^2}{p},\frac{r}{p}p_c}=0.
$$

With this information, checking some of our claims is an elementary exercise. For instance, the following line of computation: 
$$
\frac{L^2}{p^2}=\frac{M}{r}+c \qquad \stackrel{P}{\longrightarrow} \qquad \frac{L^2}{r^2}=\frac{M p}{r^2}+c \qquad \Rightarrow \qquad Mp= -c r^2+L^2.
$$
proves that the pedal of a (focal) conic is, indeed, a circle. Or, inversely, taking for granted that the pedal of a central rectangular hyperbola ($\frac{a^2 b^2}{p^2}=r^2$) is the Lemniscate of Bernoulli, we can in no time produce a pedal equation for it:
$$
\frac{a^2 b^2}{p^2}=r^2 \qquad \stackrel{P}{\longrightarrow} \qquad \frac{a^2 b^2}{r^2}=\frac{r^4}{p^2} \qquad \Rightarrow \qquad r^3=ab p.
$$

\subsubsection{Dual parallel curves}
Another important transform is \textit{dual parallel curves} $E^\star_\alpha$, i.e. a transform such that  any given point ${\bf x}$ on a curve is transformed to point ${\bf \tilde x}$ given by
$$
E^\star_\alpha:\qquad {\bf \tilde x}:=\frac{{\bf x}}{1+\alpha r}.
$$
(For the reason why this transform is called dual parallel see \cite{Blaschke6}.)

It is a stimulating exercise to verify that in pedal coordinates its action is given as follows:
\begin{equation}\label{dualtr}
f\zav{\frac{1}{p^2}-\frac{1}{r^2},\frac{1}{r}}=0 \qquad\stackrel{E^\star_\alpha}{\longrightarrow}\qquad  f\zav{\frac{1}{p^2}-\frac{1}{r^2},\frac{1}{r}-\alpha}=0.
\end{equation}

It is easy to see that the dual parallel of a line is a focal conic:
$$
\frac{1}{p^2}=c \qquad\stackrel{E^\star_\alpha}{\longrightarrow}\qquad  \frac{1}{p^2}=\frac{2\alpha}{r}+c-\alpha^2,
$$
and that focal conics are preserved by this operation since it holds:
$$
E^\star_\alpha E^\star_\beta=E^\star_{\alpha+\beta},
$$
i.e. $E^*_\alpha$ is a group of transformations (with the identity being $E^\star_0$).

This transformation appears in a generalization of Newton's theorem of revolving orbits \cite{Newton} by Mahomed and Vawda in 2000 \cite{Mahomed}. They proved that if a curve is given as a solution to the central force problem $F(r)$ changing the radial distance $r$, and the angle $\varphi$ on its every point according to rule:
\begin{equation}\label{Mahomedtr}
r\to \frac{ar}{1-br},\qquad \varphi \to \frac{1}{k}\varphi,
\end{equation}
where $a,b$ are given constants, is equivalent to changing the force as follows:
\begin{equation}\label{Mohamedeq}
F(r)\to \frac{a^3}{(1-br)^2}F\zav{\frac{ar}{1-br}}+\frac{L^2}{m r^3}(1-k^2)-\frac{bL^2}{mr^2},
\end{equation}
where $m$ is the particle mass and $L$ its angular momentum. (Interested reader should also see \cite{Bell1},\cite{Bell2}.)

\subsubsection{Square root transform}
Dealing with planar curves only allow us to describe a point on a curve ${\bf x}$ not as a vector, but as a complex number:
$$
{\bf x}=(x,y),\qquad z:= x+\ii y.
$$
With this notation we can condense both pedal coordinates $p,p_c$ into a single quantity:
$$
\frac{z \dot {\bar z}}{\abs{\dot z}}=p_c+\ii p.
$$

We are now ready to introduce the very useful \textit{square root transform} $M_\frac12$ which maps a point $z$ to its square root $\sqrt{z}$:
$$
M_\frac12: \qquad \tilde z:=\sqrt{z}.
$$
This transform translates into pedal coordinates very easily:
\begin{align*}
\tilde r&:=\abs{\tilde z}=\abs{\sqrt{z}}=\sqrt{r}.\\
\dot{\tilde z}&=\frac12 z^{-\frac12}\dot z.\\
\abs{\dot{\tilde z}}&= \frac12 r^{-\frac12} \abs{\dot z}.\\
\tilde p_c+\ii\tilde p&:=\frac{\tilde z\overline{\dot{\tilde z}}}{\abs{\dot{\tilde z}}}=\frac{z^\frac12 \frac12 \bar z^{-\frac12}\dot z}{\frac12r^{-\frac12}\abs{\dot z}}= r^{-\frac12}\frac{z\dot{\bar z}}{\abs{\dot z}}=r^{-\frac12}(p_c+\ii p).
\end{align*}
This yields:
\begin{align}\label{roottr}
f(p,r,p_c)&=0, & &\stackrel{M_{\frac12} }{\longrightarrow} & f\zav{\frac{p}{\sqrt{r}}, \sqrt{r},\frac{p_c}{\sqrt{r}}}&=0.
\end{align}

This transform maps a central conic into a focal one:
$$
ab\zav{\frac{1}{p^2}-\frac{1}{r^2}}=-\frac{(r^2-a)(r^2-b)}{r^2}, \qquad \stackrel{M_{\frac12} }{\longrightarrow} \qquad ab\zav{\frac{1}{p^2}-\frac{1}{r^2}}=-\frac{(r-a)(r-b)}{r^2}.
$$
See (\ref{focalconic}), (\ref{centralconic}).

It also maps a circle into a central Cassini oval, which is the locus of points such that the product of distances from two foci is constant:
$$
\abs{z-a^2}=R, \qquad \stackrel{M_{\frac12} }{\longrightarrow} \qquad \abs{\tilde z^2-a^2}=R,\qquad \Rightarrow \qquad \abs{\tilde z-a}\abs{\tilde z+a}=R.
$$
This gives us a rather nice pedal equation for a central Cassini oval:
$$
2pR= r^2+R^2-\abs{a}^2\qquad \stackrel{M_{\frac12} }{\longrightarrow} \qquad
2p R= (r+ R^2-\abs{a}^2)\sqrt{r}.
$$

%
\subsubsection{Rotational frame of reference} 
The dynamical system (\ref{dynsys}) is perhaps the best way how to actually \textit{draw} a curve given in pedal coordinates -- solving the second order differential equation, in fact, avoids many numerical instabilities.

Let us therefore reverse the logic:

\begin{corollary}
A curve ${\bf x}$ given in pedal coordinates:
$$
 G(r)^2\zav{\frac{1}{p^2}-\frac{1}{r^2}}=F(r),
$$
where $F,G$ are smooth functions on $(0,\infty)$, solves the following Cauchy problem:
\begin{equation}\label{dynsys2}
{\bf \ddot x}=\frac{1}{2r}\zav{F(r)+\frac{G(r)^2}{r^2}}'{\bf x}+\frac{G'(r)}{r}{\bf {\dot x}}^\perp, \qquad \abs{\bf \dot x_0}^2=F(r_0)+\frac{G(r_0)^2}{r_0^2},\quad {\bf \dot x_0}\cdot {\bf x_0}^\perp =G(r_0). 
\end{equation}
\end{corollary}
\begin{proof}
Exercise.
\end{proof}

For a dynamical system of the form (\ref{dynsys2}) there is a very natural transform that is done routinely in the field of celestial mechanics -- passing to a rotational frame of reference with an angular velocity $\omega$.

Specifically, we define the point ${\bf \tilde x}$ in RFR as follows: 
$$
R_\omega:\qquad {\bf \tilde x}:= \cos(\omega t) {\bf x}+\sin(\omega t){\bf x}^\perp.
$$ 

Equations (\ref{dynsys2}) are transformed as follows:
$$
{\bf \ddot{ \tilde x}}=\frac{1}{2r}\zav{F(r)+\frac{G(r)^2}{r^2}}'{\bf \tilde x}+\frac{G'(r)}{r}{\bf {\dot{ \tilde x}}}^\perp+\omega^2 {\bf \tilde x}-2\omega {\bf \dot{\tilde x}}^\perp,
$$
i.e. centrifugal $\omega {\bf \tilde x}$ and Coriolis $-2\omega {\bf \dot{\tilde x}}^\perp$ terms are added.

The transformed initial condition reads:
$$
\abs{{\bf \dot{\tilde x}_0}}^2=F(r_0)+2G(r_0)\omega +\omega^2 r_0^2+\frac{G(r_0)^2}{r_0^2},\qquad {\bf \dot{ x}_0}\cdot {\bf x_0}^\perp=G(r_0)+\omega r_0^2.
$$

The resulting Cauchy problem can be translated back into pedal coordinates with the following net effect:
\begin{equation}\label{Rw} 
 G(r)^2\zav{\frac{1}{p^2}-\frac{1}{r^2}}=F(r)\qquad \stackrel{R_\omega}{\longrightarrow}\qquad  \zav{G(r)+\omega r^2}^2\zav{\frac{1}{p^2}-\frac{1}{r^2}}=F(r).
 \end{equation}

\begin{example}
Starting with a central ellipse (\ref{centralconic}) with axis $a,b$:
$$
a^2b^2\zav{\frac{1}{p^2}-\frac{1}{r^2}}=\frac{(a^2-r^2)(r^2-b^2)}{r^2},
$$
we obtain in RFR:
\begin{equation}\label{epicycloids}
\stackrel{R_\omega }{\longrightarrow}\qquad (ab+\omega r^2)^2\zav{\frac{1}{p^2}-\frac{1}{r^2}}=\frac{(a^2-r^2)(r^2-b^2)}{r^2}.
\end{equation}
The solutions correspond exactly to \textit{Epicycloids}, which are curves traced by a point on a rotating circle whose center is also rotating on a different circle (whose center is at origin).  

This can be easily verified by solving the corresponding dynamical systems:
$$
{\bf \ddot x}=-{\bf x}\qquad \stackrel{R_\omega}{\longrightarrow}\qquad {\bf \ddot x}= (\omega ^2-1){\bf x}-2\omega {\bf \dot x}^\perp.
$$ 
\end{example}
\begin{example}
Consider a dynamical system:
$$
{\bf \ddot x}=-\frac{M}{r^3}{\bf x}+\frac{F}{r}{\bf x}-\omega^2 {\bf x},
$$
corresponding to the evolution of a test particle that is not only under the gravitational influence of a point mass (the $-M/r^3\,{\bf x}$ term), but also subject to gravitational attraction of a spherically symmetric halo of dark matter ($-\omega^2 {\bf x}$), and in the presence of dark energy (i.e. constant outward repulsing force $+F/r\,{\bf x}$ ).
  
In the pedal coordinates this takes the form
$$
\frac{L^2}{p^2}=\frac{2M}{r}+2Fr-\omega^2 r^2+c.
$$
Or
$$
L^2\zav{\frac{1}{p^2}-\frac{1}{r^2}}=\frac{-\omega^2 r^4+2Fr^3+cr^2+2M r-L^2}{r^2}.
$$
The polynomial 
$$
h(r):=-\omega^2 r^4+2Fr^3+cr^2+2M r-L^2,
$$
has at least two positive real roots $r_0,r_1$, we may even assume $r_0\leq r\leq r_1$ -- otherwise the expression above does not represent a curve. There are also two additional roots $r_2,r_3$ that are not necessarily positive nor real (but its product $r_2r_3$ must be real).

Hence the pedal equation can also be written:
\begin{equation}\label{DKPr}
r_0r_1r_2r_3\zav{\frac{1}{p^2}-\frac{1}{r^2}}=\frac{(r-r_0)(r_1-r)(r-r_2)(r-r_3)}{r^2},\qquad r_0\leq r\leq r_1,
\end{equation}
where
\begin{align*}
2F&=\omega^2\zav{r_0+r_1+r_2+r_3},\\
c&=-\omega^2\zav{r_0r_1+r_0r_2+r_0r_3+r_1r_2+r_1r_3+r_2r_3},\\
2M&=\omega^2\zav{r_0r_1r_2+r_0r_1r_3+r_0r_2r_3+r_1r_2r_3},\\
L^2&=\omega^2 r_0r_1r_2r_3.\\
\end{align*}

In \cite{Blaschke6} it was shown that provided $r_0r_3=r_1r_2$ (which happens when $L^2F^2=M^2\omega^2$) going to the RFR with angular speed $\omega$ (i.e. eliminating dark matter)  the resulting curve is precisely a Cartesian oval. 

Recall that the Cartesian oval is the locus of points such that a linear combination of distances from two foci is constant. 
Specifically, it is a curve given by:
$$
\abs{\bf x}+\alpha\abs{\bf x-a}=C.
$$
And the pedal equation:
$$
\zav{r_0r_3-r^2}^2\zav{\frac{1}{p^2}-\frac{1}{r^2}}=- \frac{(r-r_0)(r-r_1)(r-r_2)(r-r_3)}{r^2}, 
$$
where 
$$ r_0=\frac{C-\alpha\abs{\bf a}}{1+\alpha}, \qquad r_1=\frac{C-\alpha\abs{\bf a}}{1-\alpha},\qquad r_2=\frac{C+\alpha\abs{\bf a}}{1+\alpha}, \qquad r_3=\frac{C+\alpha\abs{\bf a}}{1-\alpha}.
$$
Indeed:
$$
r_0r_3=r_1r_2.
$$
\end{example}
\begin{example}
The transform $R_\omega$ is highly dependent on the form of the pedal equation in question. For instance, take once again a central ellipse
$$
r_0^2r_1^2\zav{\frac{1}{p^2}-\frac{1}{r^2}}=\frac{(r^2-r_0^2)(r_1^2-r^2)}{r^2},
$$
which corresponds to solutions of Hooke's law
\begin{equation}\label{ds1}
{\bf \ddot x}=-{\bf x}.
\end{equation}

By a simple algebraic operation like multiplying both sides by a factor $(b r^2+a)^2$
$$
r_0^2r_1^2(b r^2+a)^2\zav{\frac{1}{p^2}-\frac{1}{r^2}}=\frac{(r^2-r_0^2)(r_1^2-r^2)(br^2+a)^2}{r^2},
$$
we obtain a different dynamical system:
\begin{equation}\label{ds2}
{\bf \ddot y}=-\zav{3b^2 r^4+2r^2\zav{2ba+b^2(r_0^2+r_1^2)}-a^2+2ba(r_0^2+r_1^2)}{\bf y}+2br_0r_1{\bf \dot y}^\perp.
\end{equation}

Equations (\ref{ds1}) and (\ref{ds2}) are, however, not equivalent. Ellipses are solutions of (\ref{ds2}) only for a very particular set of initial condition. 

But even if some orbits of  (\ref{ds1}), (\ref{ds2}) are the same the corresponding solutions still differ by a parametrization. The ellipse ${\bf x}$ is parameterized so that equal areas are swept by equal time (Kepler's law),
but ellipse ${\bf y}$ is not -- in fact, its parametrization is so that equal areas are swept in equal time on a different coaxial ellipse. 
\subsubsection{Ellipse in Generalized rotation frame of reference (GRFR)}
\begin{center}
\includegraphics[scale=0.3,trim=5cm 1cm 10cm 1.5cm,clip]{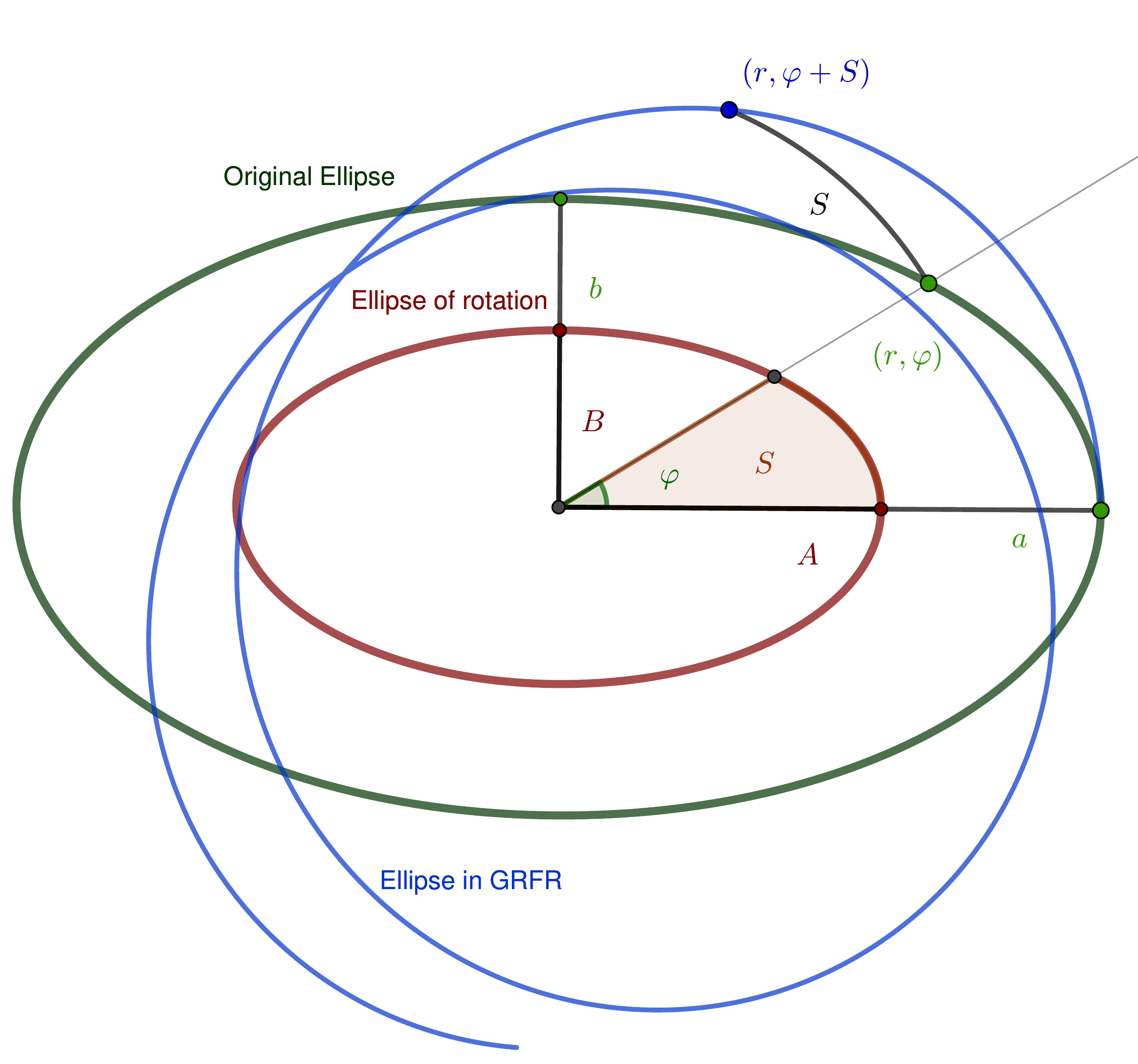}
\end{center}

Consider a central ellipse with semi-axes $A,B$ (red in the picture). A point on it is given by $x=A\cos t$, $y=B\sin t$. The area $S$ swept by a radial vector up to that point is
$$
S=\frac12\inte{0}{t} (-y\dot x+x\dot y)\dd t=\frac12\inte{0}{t}\zav{BA\sin^2 t+AB\cos^2 t}\dd t=\frac{AB t}{2}.
$$
Consider now a second coaxial ellipse (green in the picture) with semi-axes $a,b$. We pick a point on it $(r,\varphi)$ in polar coordinates with the same polar angle $\varphi$ as has the point $(x,y)$, i.e. 
$$
\tan \varphi=\frac{y}{x}=\frac{B}{A}\tan t,\qquad \Rightarrow\qquad \varphi=\arctan\zav{\frac{B}{A}\tan t}.
$$ 
Thus
$$
\cos \varphi=\frac{A\cos t}{\sqrt{A^2\cos^2 t+B^2\sin^2 t}},\qquad \sin \varphi=\frac{B\sin t}{\sqrt{A^2\cos^2 t+B^2\sin^2 t}}.
$$

The distance $r$ can be computed from the equation of central ellipse as follows:
$$
r=\frac{ab}{\sqrt{b^2\cos^2 \varphi+a^2\sin^2 \varphi}}=\frac{ab\sqrt{C^2\cos^2 t+B^2}}{\sqrt{D\cos^2 t+a^2 B^2}},
$$
where 
$$
C^2:=A^2-B^2,\qquad c^2:=a^2-b^2, \qquad D:=b^2A^2-a^2B^2.
$$
We now rotate the point $(r,\varphi)$ around origin by the amount of $\pm S$ (depending whether clockwise or anticlockwise). The new point (on a blue curve in the picture) has coordinates
$$
\tilde r:=r=\frac{ab\sqrt{C^2\cos^2 t+B^2}}{\sqrt{D^2\cos^2 t+a^2 B^2}},\qquad 
\tilde \varphi:=\varphi\pm S=\arctan\zav{\frac{B}{A}\tan t}\pm \frac12 t AB.
$$

Now $$
\zav{\frac{\partial \tilde r}{\partial\tilde \varphi}}^2=\zav{\frac{\dot {\tilde r}}{\dot{\tilde \varphi}}}^2=
\frac{(\tilde r^2-b^2)(a^2-\tilde r^2)\zav{C^2a^2b^2-D^2 \tilde r^2}^2}{(ab)^2\zav{\zav{\pm\frac12A^2B^2c^2-D}\tilde r^2+(ab)^2C^2}^2}.
$$

But, generally,
$$
\zav{\frac{ \partial r}{r \partial\varphi}}^2=\zav{\frac{r \dot r}{r^2 \dot\varphi}}^2=\zav{\frac{{\bf  x}\cdot {\bf \dot x}}{{\bf \dot x}\cdot {\bf x}^\perp}}^2=\zav{\frac{p_c}{p}}^2=r^2\zav{\frac{1}{p^2}-\frac{1}{r^2}}.
$$

Thus the resulting curve -- \textit{ellipse in GRFR} -- is given in pedal coordinates as follows: 
\begin{equation}\label{eGRFR}
(ab)^2\zav{\zav{ \pm\frac12(AB)^2c^2-D}r^2+C^2(ab)^2}^2\zav{\frac{1}{p^2}-\frac{1}{r^2}}=\frac{(a^2-r^2)(r^2-b^2)\zav{C^2(ab)^2-Dr^2}^2}{r^2},
\end{equation}
where the linear eccentricities $c,C$ and the combined eccentricity $D$ are given:
$$
C^2:=A^2-B^2,\qquad c^2:=a^2-b^2, \qquad D:=b^2A^2-a^2B^2.
$$

In other words, the ellipse in GRFR is a central ellipse parameterized so that equal areas are swept in equal times on different co-centric and co-axial ellipse observed form the rotation frame of reference.

Notice that applying the transform $R_\omega$ with angular speed $\omega=\mp\frac12 (AB)^2 c^2$ we obtain the original ellipse:
$$
(\ref{eGRFR})\qquad \stackrel{R_{\omega}}{\longrightarrow}\qquad 
(ab)^2\zav{\frac{1}{p^2}-\frac{1}{r^2}}=\frac{(a^2-r^2)(r^2-b^2)}{r^2}.
$$ 
So this curve looks complicated but in the correct rotational frame it is just a central ellipse.

\begin{figure}[h] \label{galeryeGRFR}
\caption{Ellipses in GRFR with $a=3,b=2,A=2$ and $B$ is equal in turn: $-1,1,3$ top row and $-\frac13,-\frac35,\frac52$ bottom row. The minus sign indicates negative, clockwise orientation.}
\centering 
\includegraphics[scale=0.9,trim=2cm 5cm 1cm 11cm,clip]{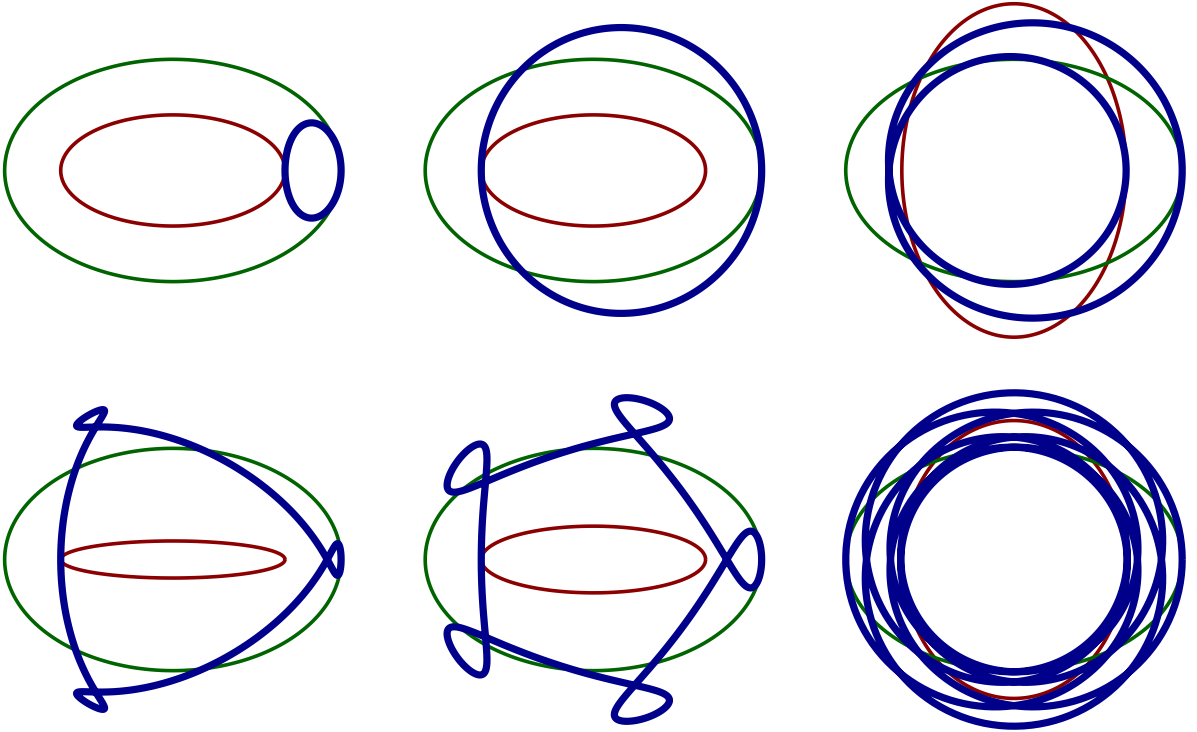}
\end{figure}
\end{example}

\section{Zero angular momentum case}\label{Sec3}
Assuming ${\bf L}=0$, our solution of free double linkage (\ref{Pedeq}) can be brought into a very simple form.
\begin{corollary}\label{Cor1} The curve traced by the node ${\bf x}$ of a coplanar free double linkage with zero total angular momentum ${\bf L}=0$ as viewed from the center of mass is given by:
\begin{equation}\label{L0casepedgen}
a^2\zav{r^2+r_0r_1}^2\zav{\frac{1}{p^2}-\frac{1}{r^2}}=\frac{(r^2-r_0^2)(r_1^2-r^2)(br^2+a)^2}{r^2},
\end{equation} 
where
\begin{align*}
a&=M m_y(r_1-r_0)^2, & b&=4m_xm_z,\\
r_1&=\frac{l_1(m_z+m_y)+ l_2m_z}{M}, & r_0&=\frac{l_1(m_z+m_y)- l_2m_z}{M}.\\
\end{align*}
\end{corollary}
\begin{proof}Trivial.
\end{proof}
Let us discuss few interesting special cases.

\subsection{Central ellipse} 
Clearly, setting $a=br_0r_1$ we can eliminate the term $(br^2+a)^2$ from both sides of (\ref{L0casepedgen})  and obtain:
\begin{equation}\label{L0centralellipse}
\zav{r_0r_1}^2\zav{\frac{1}{p^2}-\frac{1}{r^2}}=\frac{(r^2-r_0^2)(r_1^2-r^2)}{r^2},\qquad a=br_0r_1,
\end{equation}  
i.e. central ellipse with semi axis $r_0,r_1$ (see (\ref{centralconic})).

However, let us say that we cannot recognize this curve to be an ellipse. In that case (which happens all the time) we can use some transforms introduced in Section \ref{Pedsec} to try to simplify the result. For instance, using the complex square root $M_\frac12$ (see (\ref{roottr})):
$$
\zav{r_0r_1}^2\zav{\frac{1}{p^2}-\frac{1}{r^2}}=\frac{(r^2-r_0^2)(r_1^2-r^2)}{r^2} \qquad \stackrel{M_{\frac12} }{\longrightarrow} \qquad \zav{r_0r_1}^2\zav{\frac{1}{p^2}-\frac{1}{r^2}}=\frac{(r-r_0^2)(r_1^2-r)}{r^2},
$$
we obtain the equation for the focal ellipse (see (\ref{focalconic})). But say that even this we cannot identify. Using dual parallel transform $E^*_{\alpha}$ (\ref{dualtr}) we get: 
$$
\zav{r_0r_1}^2\zav{\frac{1}{p^2}-\frac{1}{r^2}}=\frac{(r-r_0^2)(r_1^2-r)}{r^2}\qquad  \stackrel{E^*_{\alpha} }{\longrightarrow} \qquad \frac{r_0^2r_1^2}{p^2}=\frac{(r_1^2-r_0^2)^2}{4r_0^2r_1^2}, \qquad \alpha:=-\frac{r_0^2+r_1^2}{2r_0^2r_1^2},
$$ 
which is clearly a line. 

Taking the inverse transforms we can thus conclude that our solution is a complex square of a dual parallel of a line -- which is perhaps the most cumbersome description of a central ellipse there is. But it contains information about the construction.

As we will see, the exact same transform, that is a square root $M_\frac12$ followed by a dual parallel transform $E^*_\alpha$ for suitable $\alpha$ will always produce something interesting. Perhaps the transform $E^*_\alpha M_\frac12$ is somewhat natural for double linkage.
\subsection{Kepler problem in general relativity}
Take, for instance, the case $r_0=0$, meaning that the point ${\bf x}$ is allowed to pass through the center of mass of the linkage.

Putting $r_0=0$ into (\ref{L0casepedgen}) we obtain
\begin{equation}\label{r0case}
a^2\zav{\frac{1}{p^2}-\frac{1}{r^2}}=\frac{(r_1^2-r^2)(br^2+a)^2}{r^4},\qquad r_0=0.
\end{equation}
Using the square root transform $M_{\frac12}$ we get
\begin{equation}\label{todarkkepler}
(\ref{r0case})\qquad  \stackrel{M_{\frac12}}{\longrightarrow}\qquad a^2\zav{\frac{1}{p^2}-\frac{1}{r^2}}=\frac{(r_1^2-r)(br+a)^2}{r^3}.
\end{equation} 

This equation is actually of the same form as the so-called ``Kepler problem in General relativity'', i.e. the relativistic correction to a classical Kepler problem. 

Kepler problem in GR  is the problem of orbits (or geodesics) around a non-rotating compact body described by the Schwarzschild solution of the Einstein equations of General relativity \cite{schwarzschild}:
\begin{equation}\label{KPGR}
{r^\prime_\varphi}^2=\frac{r^4}{b_s^2}-\zav{1-\frac{r_s}{r}}\zav{\frac{r^4}{a_s^2}+r^2},
\end{equation}
where
$$
r_s:=\frac{2G M}{c^2}, \qquad a_s:=\frac{L}{GM c}, \qquad b_s:=\frac{cL}{E}.
$$
The quantity $r_s$ is the Schwarzschild radius, $L$ is  the specific angular momentum and  $E$ is the specific energy of a test particle.

As shown in \cite{Blaschke6}, the orbits are given in pedal coordinates as follows:
\begin{equation}\label{KPGRp}
(\rho_0+\rho_1+\rho_2)\zav{\frac{1}{p^2}-\frac{1}{r^2}}=\frac{(r- \rho_0)(r- \rho_1)(r- \rho_2)}{r^3},
\end{equation}
where
\begin{align*}
\rho_0\rho_1\rho_2&=\frac{r_s a^2_sb^2_s}{b_s^2-a_s^2}.\\
\rho_0\rho_1+\rho_0\rho_2+\rho_1\rho_2&=\frac{a^2_sb^2_s}{b_s^2-a^2_s}.\\
\rho_0+\rho_1+\rho_2&=\frac{r_sb_s^2}{b_s^2-a_s^2}.
\end{align*}

Choosing $\rho_0=r_1^2$, $\rho_1=\rho_2=-a/b$ we get
$$
b(2a-b r_1^2)\zav{\frac{1}{p^2}-\frac{1}{r^2}}=\frac{(r_1^2- r)(br+a)^2}{r^3},
$$
almost a match with (\ref{KPGRp}). To get precise equality we must perform on (\ref{todarkkepler}) the dual parallel transform $E^\star_{\alpha}$ with $\alpha=2b/(3a)$. 
  
The two problems, i.e. $r_0=0$ case of zero angular momentum free double linkage (\ref{r0case}) and Kepler problem in general relativity (\ref{KPGRp}) comes obviously from very different origins. And yet, there is a connection. This connection is almost certainly just a mathematical coincidence, but it is kind of appealing to think that a movement of a double linkage is somehow connected with the movement inside a Black Hole. (Remember $r_0=0$.) 
\subsection{Dark Kepler problem} There is a general way how to get rid of the magnetic term $(r+r_0r_1)^2$ in (\ref{L0casepedgen}) -- namely using the transform $E^\star_{\alpha} M_{\frac12}$.

First the square root transform:
$$
(\ref{L0casepedgen})\qquad \stackrel{M_{\frac12}}{\longrightarrow}\qquad a^2\zav{r+r_0r_1}^2\zav{\frac{1}{p^2}-\frac{1}{r^2}}=\frac{(r-r_0^2)(r_1^2-r)(br+a)^2}{r^2}.
$$
Rewrite the result as follows:
$$
a^2\zav{1+\frac{r_0r_1}{r}}^2\zav{\frac{1}{p^2}-\frac{1}{r^2}}=\zav{1-\frac{r_0^2}{r}}\zav{\frac{r_1^2}{r}-1}\zav{b+\frac{a}{r}}^2,
$$
we can easily apply a dual parallel transform: 
$$
\stackrel{E^\star_{-\frac{1}{r_0r_1}}}{\longrightarrow}\qquad 
a^2\zav{\frac{r_0r_1}{r}}^2\zav{\frac{1}{p^2}-\frac{1}{r^2}}=\zav{1-\frac{r_0^2}{r}+\frac{r_0}{r_1}}\zav{\frac{r_1^2}{r}-\frac{r_1}{r_0}-1}\zav{b+\frac{a}{r}-\frac{a}{r_0r_1}}^2.
$$
Tidying this up we get
\begin{equation}\label{DKP}
a^2 (r_0r_1)^5\zav{\frac{1}{p^2}-\frac{1}{r^2}}=\frac{\zav{(r_0+r_1)r-r_1r_0^2}\zav{r_1^2r_0-(r_0+r_1)r}\zav{(br_0r_1-a)r+ar_0r_1}^2}{r^2},
\end{equation}
which is exactly the Dark Kepler problem (\ref{DKPr}). Or rather its special case with a double root. Thus we have discovered a particular solution.
 
\subsection{General construction}

\begin{corollary}\label{Cor3} Provided $r_0\not=0$, curves given by (\ref{L0casepedgen}) are ellipses in GRFR with  $r_1,r_0$ being the semi-axes of the original ellipse and with the ellipse of rotation given by
$$ 
A=\sqrt{\abs{\frac{2r_1}{r_0}\frac{{a-br_0r_1}}{br_0^2+a}}},\qquad B=\sqrt{\abs{\frac{2r_0}{r_1}\frac{{a-br_0r_1}}{br_1^2+a}}}.
$$
\end{corollary}
\begin{proof}
Multiplying the equation (\ref{L0casepedgen}) by some nonzero number $\lambda^2$ and comparing with (\ref{eGRFR}) -- where we set $a=r_1,b=r_0$ we obtain three equations in three unknowns $(A,B,\lambda)$:
\begin{align*}
\lambda a&=r_0r_1\zav{\pm\frac12 (AB)^2(r_1^2-r_0^2)-D},\\
\lambda a&=(A^2-B^2)(r_0r_1)^2,\\
\lambda b&=-D=r_1^2 B^2-r_0^2 A^2.
\end{align*}
The reader can easily verify that the desired $A,B$ solves this system together with 
$$
\lambda=\frac{2r_0r_1(r_1^2-r_0^2)\zav{a-br_0r_1}}{(a+r_1^2 b)(a+r_0^2 b)}.
$$
The $\pm$ sign in the first equation must be chosen so that a solution exists and it indicates clockwise (-) or anti-clockwise (+) orientation of the rotation.
\end{proof}
Note that even though solutions of ${\bf L}=0$ double linkage are ellipses in GRFR, the reverse is not true. Not all ellipses in GRFR solve ${\bf L}=0$ double linkage. The following connection between the original ellipse and the ellipse of rotation must hold:
$$
\frac{r_1-r_0}{2}=\frac{r_1^2}{r_0A^2}-\frac{r_0^2}{r_1 B^2}.
$$

\section{Conclusion}
As we saw, the movement of a free double linkage is far from trivial, even if we restrict ourselves to the case of zero total angular momentum ${\bf L}=0$.
 
We have demonstrated that solutions can be linked to both the Kepler problem in general relativity and to the Dark Kepler problem. Furthermore, we have fully described the resulting orbits as ellipses in generalized rotation frame of reference. 

The case  ${\bf L}\not=0$ for which we were able to provide algebraic closed solution in pedal coordinates (\ref{Pedeq}) is more challenging to fully describe and further research is needed. 

\section{Acknowledgment}
The authors thank the anonymous referee and Evariste Boj for careful reading of the manuscript and suggesting numerous improvements.
   
F. Blaschke and M. Blaschke would like to express their acknowledgment for the institutional support of the Research Centre for Theoretical Physics and Astrophysics, Institute of Physics, Silesian University in Opava.

P. Blaschke was supported by the GA\v CR grant no. 21-27941S and RVO funding 47813059.

\section{Data availability}
Data sharing is not applicable to this article as no new data were created or analyzed in this
study.

\end{document}